\documentclass[11pt,a4paper]{article}

\usepackage[T1]{fontenc}
\usepackage[utf8]{inputenc}
\usepackage{amsmath,amssymb,amsthm,mathtools}
\usepackage{booktabs}
\usepackage{array}
\usepackage{enumitem}
\usepackage{float}
\usepackage[section]{placeins}
\usepackage{caption}
\usepackage[margin=2.9cm]{geometry}
\usepackage[expansion=false]{microtype}
\usepackage{tikz}
\usetikzlibrary{arrows.meta,positioning}
\usepackage[hidelinks]{hyperref}
\newtheorem{theorem}{Theorem}[section]
\newtheorem{lemma}[theorem]{Lemma}
\newtheorem{proposition}[theorem]{Proposition}
\newtheorem{corollary}[theorem]{Corollary}
\theoremstyle{definition}
\newtheorem{definition}[theorem]{Definition}

\theoremstyle{remark}
\newtheorem{remark}[theorem]{Remark}

\newcommand{\F}{\mathbb{F}}
\newcommand{\gap}[2]{\langle #1,#2\rangle}
\newcommand{\igap}[3]{\langle #1,#2\rangle@{#3}}
\newcommand{\Gaps}{\Gamma}
\newcommand{\uni}{\mathcal{U}}
\newcommand{\kap}{\kappa}
\newcommand{\acc}{\mathsf{A}}
\newcommand{\rel}{\mathsf{R}}
\newcommand{\nuf}{\nu}
\newcommand{\Live}{\mathcal{L}}
\newcommand{\op}[1]{\textsf{#1}}
\newcommand{\ans}[1]{\textit{#1}}
\newcommand{\rl}[1]{\textnormal{(N#1)}}
\newcommand{\Maint}{\mathcal{M}}
\newcommand{\Audit}{\mathcal{V}}
\newcommand{\Births}{\mathcal{B}}
\newcommand{\Consumed}{\mathcal{C}}
\newcommand{\I}{\mathcal{I}}
\newcommand{\supp}{\operatorname{supp}}
\newcommand{\outdeg}{\operatorname{out}}
\newcommand{\indeg}{\operatorname{in}}

\title{\bfseries Split Tallies: A Discrete Certificate Calculus\\ for Auditing Dynamic Ordered Sets in Constant Memory}

\author{Faruk Alpay\thanks{Correspondence: \texttt{alpay@lightcap.ai}}
\and Levent Sar{\i}o\u{g}lu}
\date{Department of Computer Engineering, Bah\c{c}e\c{s}ehir University, Istanbul, Turkey\\
\texttt{\{faruk.alpay, levent.sarioglu\}@bahcesehir.edu.tr}\\[1.2em]
\today}

\begin{document}
\maketitle

\begin{abstract}
A dynamic ordered set is maintained by one party and consulted by another.
The maintainer executes insertions, deletions, membership tests, predecessor
and successor queries, and extremum queries, and announces an answer for each;
the answers are used immediately and verified never.  We ask how little an
auditor needs in order to certify, after the fact, that \emph{every} announced
answer was correct.  The auditor we construct is passive: it watches the
operation--answer stream, requires the maintainer to append a constant number
of machine words per operation to a public, append-only \emph{tally}, and keeps
a private \emph{stock} of five words and a flag.  At audit time the maintainer
\emph{discloses} its claimed live intervals; the auditor accepts or rejects in
one pass.  Honest maintainers are accepted with probability one.  A maintainer
that misreported even a single answer among $T$ operations is accepted with
probability at most $(4T+1)/p$ over the auditor's one secret field element,
where $p$ may be taken as a $61$-bit prime, with no assumption on the
maintainer's computational power.  The certificate calculus rewrites the
\emph{gaps} of the ordered set---its maximal vacant intervals---under an
\emph{indenture discipline} that timestamps every gap by a public counter.
A discrete normal-form theorem identifies these gaps with the unique unit path
in the interval acyclic graph, equivalently with an incidence vector of
boundary $\mathbf{1}_0-\mathbf{1}_{U+1}$; insertions and deletions are then
exactly the one-vertex refinements and coarsenings of that path.  Both pillars
of the construction are provably indispensable: any deterministic
auditor with end-of-session soundness must keep at least
$2n-\log_2(2n+1)$ bits of state, the same floor binds every perfectly complete
auditor whose coin tosses are visible to the maintainer, and removing the
timestamp rule admits an explicit three-operation forgery that achieves exact
balance and is accepted with probability one.  A leaf-oriented $(2,4)$-tree
realizes the maintainer's side in $O(\log n)$ worst-case time and one extra
word per element, and a potential-function argument shows that its splits,
fusions, shares, and root events number at most $2m$ over any $m$ updates---a
bound the auditor itself can certify by counting one-word attestations.
Audits may be placed at arbitrary checkpoints at cost linear in the current
size, and accepted epochs compose with additive error.  All constants are
explicit, and a complete six-operation session, together with its forged twin,
is worked to the last residue.
\medskip

\end{abstract}

\section{Introduction}\label{sec:intro}

A growing share of the data structures that programs rely on are no longer
executed where they are trusted.  The ordered dictionary at the heart of an
index, a scheduler, or a routing table may live in a procured library, behind
a remote interface, on a rented machine, or inside an accelerated component
whose internals are opaque to its caller.  The caller issues operations,
receives answers, and acts on them at once.  Correctness, when it is
established at all, is established by testing the component in isolation---a
guarantee about \emph{some} executions, not about \emph{this} one.

This paper develops the opposite guarantee for the dynamic ordered set: a
certificate that travels with the execution itself.  We place between the
maintainer of the structure and the consumer of its answers a third party of
deliberately minimal power, the \emph{auditor}.  The auditor never touches the
structure, never sends a message, and stores five machine words and a flag.
It watches the stream of operations and announced answers, and it obliges the
maintainer to append, for each operation, a short record---a
\emph{notch}---to a public, append-only sequence we call the \emph{tally}.
At a moment of the auditor's choosing the maintainer must \emph{disclose} the
intervals it claims are currently vacant, and the auditor renders a verdict.
If every announced answer was correct, the verdict is \emph{accept}, always.
If even one answer was wrong, the verdict is \emph{reject}, except with
probability at most $(4T+1)/p$ over a single field element the auditor drew in
private at the start; with a $61$-bit prime this is below $2^{-38}$ for any
session of up to a million operations.  No bound is assumed on the
maintainer's computational resources, and no key, signature, or hardness
assumption appears anywhere in the construction.

The design borrows its vocabulary, and its architecture, from a very old
auditing instrument.  For roughly six centuries the English Exchequer recorded
debts on wooden tally sticks: the amount was cut as notches across the stick,
the stick was split lengthwise through the notches, the creditor kept one
half---the \emph{stock}---and the debtor the other---the \emph{foil}---and at
settlement the halves were laid together; only the genuine counterpart would
mesh, grain and notch, with its twin (see, e.g.,
Jenkinson~\cite{jenkinson1925} and Baxter~\cite{baxter1989}).  The scheme
below is a split tally for order semantics.  The tally is the public notch
stream; the foil is the bookkeeping the maintainer carries inside its
structure; the stock is the auditor's five private words; the audit is the
meshing of the halves.  None of the parts is useful alone, and a forged half,
we prove, meshes with probability at most $(4T+1)/p$.

\subsection{What is certified, and against whom}

The object under audit is a set $S$ over a finite ordered universe,
manipulated through a totalized interface: \op{insert}, \op{delete},
\op{member}, \op{pred}, \op{succ}, \op{min}, \op{max}, each returning an
answer whose correct value is a function of the operation sequence alone
(Section~\ref{sec:model}).  The maintainer is adversarial in the strongest
sense: it may answer and notch arbitrarily, adaptively, with unbounded
computation, and it may even choose the operation sequence itself.  The
auditor is passive and silent until the audit; in particular the maintainer
learns nothing about the auditor's secret during the run.  Soundness is
retrospective: a false answer is not interrupted at the instant it is uttered,
but it is, with the stated probability, exposed at the audit, and audits may
be placed as densely as desired (Section~\ref{sec:epochs}).

The technical core is a calculus of \emph{gaps}.  A set $S$ over the
augmented universe partitions the order into $|S|+1$ maximal vacant
intervals, and every interface operation acts on this family by a local
rewrite: an insertion splits one gap in two, a deletion merges two adjacent
gaps into one, and every query is answered by exhibiting a single gap
(Section~\ref{sec:gaps}).  Section~\ref{sec:normalforms} isolates the
underlying discrete object: the current gap family is the unique directed
unit path from $0$ to $U+1$ in the interval acyclic graph, and the split and
merge rules are precisely its one-vertex refinements and coarsenings.  The
auditor therefore certifies not the set but the evolving path of gaps.  It
does so with two product accumulators in a prime field---one absorbing every
gap \emph{birth}, one absorbing every gap \emph{consumption}---and a public
counter that stamps each birth with a fresh \emph{code}.  The citation rules, which we call the \emph{indenture
discipline}, force every operation to consume the gaps it rewrites, citing
their codes, and force every answer to be the unique value determined by the
citation (Section~\ref{sec:discipline}).  Acceptance reduces to one equality
between accumulators, and the soundness analysis to unique factorization in
$\F_p[z]$ (Section~\ref{sec:audit}).

Mechanism design of this kind earns its keep only if its parts are necessary,
so Section~\ref{sec:necessity} is devoted to impossibility.  First, no
deterministic auditor can do this job in small space: end-of-session
soundness for sets of size $n$ over a universe of size $2n$ forces at least
$\log_2\binom{2n}{n}\ge 2n-\log_2(2n+1)$ bits of auditor state
(Theorem~\ref{thm:floor}), so randomness is not a convenience but a
requirement.  Second, the randomness must be \emph{hidden}: every perfectly
complete auditor whose coins are visible to the maintainer inherits the same
floor, with soundness error one below it
(Corollary~\ref{cor:visible}).  Third, the timestamp rule is load-bearing on
its own: if the auditor checks everything else but accepts citations of
not-yet-issued codes, an explicit three-operation strategy forges a
membership answer, restores \emph{exact} multiset balance---not merely a
collision---and is accepted with probability one; the same example shows that
end-of-session meshing cannot substitute for the temporal rule
(Proposition~\ref{prop:replay}).

The maintainer's side must also be realizable without penalty, and
Section~\ref{sec:maintainer} supplies a concrete one: a leaf-oriented
$(2,4)$-tree whose leaves are doubly linked and carry one extra word each---%
the code of the gap to the right.  Every operation runs in $O(\log n)$
worst-case time, emits its notch in $O(1)$ additional work, and never touches
a code during rebalancing.  The rebalancing itself is then made auditable:
we prove, by an explicit potential function with constants
$\varphi(1{:}5)=(3,1,0,2,4)$, that splits, fusions, shares, and root events
number at most $2m$ over any $m$ updates from an empty structure
(Theorem~\ref{thm:envelope}).  If the maintainer appends a one-word
\emph{attestation} per structural event, the auditor certifies the
amortized-cost claim by counting---amortization with a public envelope.
Section~\ref{sec:epochs} composes audits into epochs with halt-on-reject
semantics and additive error, and Section~\ref{sec:exemplar} executes a full
six-operation session and its forged twin numerically, every residue shown.

\subsection{Relation to prior work}\label{sec:related}

The fingerprinting mechanics descend directly from the offline memory
checkers of Blum, Evans, Gemmell, Kannan, and Naor~\cite{begkn94}, who
certified a random-access memory---an address-to-value map---with an
$\varepsilon$-biased multiset hash and write timestamps, in the
program-checking tradition of Blum and Kannan~\cite{blumkannan95}.  We
transplant the invariant from address semantics to \emph{order} semantics.
The transplant is not a relabeling: a memory checker certifies that the value
read at an address is the value last written there, whereas the queries
here---absence with its bracketing neighbors, predecessor, successor,
extrema---are properties of the \emph{arrangement} of keys, not of any stored
cell.  The gap calculus of Section~\ref{sec:gaps} is what converts
arrangement into a checkable ledger: adjacency becomes an object that can be
born, cited, and consumed, the unique-citation arguments of
Section~\ref{sec:audit} replace last-write reasoning, and the answers
themselves become forced functions of the citations.  On top of this sit
results with no memory-checking counterpart: the deterministic and
visible-coin floors of Section~\ref{sec:necessity}, the balance-restoring
replay forgery, the audited rebalancing envelope, and closed-form constants
throughout.

Certifying algorithms~\cite{mcconnell11} ask a program to emit a witness that
a \emph{single} output is correct, checkable by a simple verifier; the LEDA
system~\cite{mehlhornnaher99} demonstrated the practice at scale.  The
present setting is the dynamic analogue with an adversarial prover: thousands
of interleaved updates and queries, a verifier that may not retain the
structure, and witnesses that must compose across time.  Authenticated data
structures in the Merkle tradition~\cite{merkle89,tamassia03} solve a
different problem---public verifiability of \emph{individual} answers against
a published digest---and pay for it with cryptographic assumptions and
logarithmic-size proofs per query; here proofs are constant-size, the
guarantee is information-theoretic, and the verdict covers the entire history
at once.  Goodrich, Kornaropoulos, Mitzenmacher, and
Tamassia~\cite{goodrich17} use the word \emph{audit} for a related but
distinct forensic goal: reconstructing what an intruder learned from a
structure, rather than certifying the structure's answers to its user.

In the annotated-data-stream model of Chakrabarti, Cormode, McGregor, and
Thaler~\cite{ccmt14}, a prover annotates a stream so that a small-space
verifier can compute a function of it, and strong lower bounds trade
annotation length against verifier space for \emph{online} soundness.  Our
auditor escapes those trade-offs by design, not by contradiction: detection
here is retrospective, settled at the audit rather than at the moment of each
answer.  Whether small-state \emph{online} checking is possible at all
without cryptography is precisely the question Naor and
Rothblum~\cite{naorrothblum09} tied to the existence of one-way functions,
with the trade-offs sharpened by Dwork, Naor, Rothblum, and
Vaikuntanathan~\cite{dnrv09}; Section~\ref{sec:conclusion} states the exact
boundary this places on any strengthening of the present results.  The
algebraic identity test underlying the audit is the classical one of
Schwartz, Zippel, and DeMillo--Lipton~\cite{schwartz80,zippel79,demillo78},
and the multiset-equality use of a secret evaluation point is in the lineage
of Wegman and Carter~\cite{wegmancarter81}.  On the maintainer's side, the
$(2,4)$-tree is the classical small-order B-tree~\cite{bayermccreight72,ahu74},
its constant amortized rebalancing was established by Huddleston and
Mehlhorn~\cite{huddlestonmehlhorn82}, and our envelope theorem packages that
phenomenon, with the credit method of~\cite{tarjan85,sleatortarjan85}, into a
form an external counter can certify; skip lists~\cite{pugh90} would serve as
an alternative maintainer with expected bounds.  Kinetic data
structures~\cite{bgh99} also maintain a family of certificates alongside a
structure, but theirs expire with the motion of the data and trigger repairs,
whereas an indentured gap is consumed only by the operation that rewrites it.

\subsection{Summary of results}

Throughout, $T$ counts operations in a session, $m$ counts updates, $n$ the
current size, $U$ the universe size, and $p$ a prime exceeding
$(2T+2)(U+2)^2$.
\begin{itemize}[leftmargin=2em]
\item \textbf{Discrete normal form} (Proposition~\ref{prop:incidence} and
Theorem~\ref{thm:normalform}).  Gap families are exactly the $0$-to-$U+1$
unit paths of the interval acyclic graph, and the update rules are the unique
one-vertex path refinements and coarsenings.  This is the combinatorial core
used by every later certificate.
\item \textbf{Audit theorem} (Theorem~\ref{thm:audit}).  The constant-memory
auditor accepts every honest session with probability one and accepts a
session containing any incorrect answer with probability at most
$(4T+1)/p$, against unbounded adaptive maintainers.
\item \textbf{Deterministic floor} (Theorem~\ref{thm:floor}).  Every
deterministic auditor with end-of-session soundness keeps at least
$2n-\log_2(2n+1)$ bits of state.
\item \textbf{Visible coins} (Corollary~\ref{cor:visible}).  Every perfectly
complete auditor below that floor whose randomness is visible to the
maintainer has soundness error one.  Secrecy of the evaluation point is
therefore necessary, not incidental.
\item \textbf{Temporal rule necessity} (Proposition~\ref{prop:replay}).
Without the timestamp check, an explicit three-operation forgery achieves
exact accumulator balance and certain acceptance; meshing at disclosure does
not detect it.
\item \textbf{Maintainer} (Theorem~\ref{thm:maintainer}).  A leaf-oriented
$(2,4)$-tree supports the discipline in $O(\log n)$ worst-case time per
operation, one extra word per element, and perfect completeness.
\item \textbf{Attestation envelope} (Theorem~\ref{thm:envelope}).  Its
structural events number at most $2m$ over any $m$ updates, and the bound is
certifiable by the auditor through one-word attestations.
\item \textbf{Checkpoints and epochs} (Theorem~\ref{thm:epochs}).  An audit at
any instant costs $O(n+1)$ field operations in one pass and $O(1)$ space;
accepted epochs chain with additive soundness error.
\end{itemize}
Section~\ref{sec:exemplar} closes the loop with a fully executed session.

\section{The model}\label{sec:model}

\subsection{States, operations, answers}

Fix an integer $U\ge 1$ and let the universe be
$\uni=\{1,2,\dots,U\}$ with the natural order.  Augment it with two
\emph{sentinels}, written $0$ and $U+1$, so that
$0<1<\dots<U<U+1$.  A \emph{state} is a subset $S\subseteq\uni$; sentinels
are never elements of a state.  The interface consists of seven
\emph{operations}, each taking at most one key $x\in\uni$ and each
\emph{totalized}: every invocation returns an answer, and the correct answer
is a function of the state, hence of the operation prefix.

\begin{definition}[Operation semantics]\label{def:semantics}
Let $S$ be the current state and $x\in\uni$.
\begin{itemize}[leftmargin=2em,itemsep=0.1em]
\item $\op{ins}(x)$: if $x\notin S$ the state becomes $S\cup\{x\}$ and the
answer is \ans{inserted}; if $x\in S$ the state is unchanged and the answer is
\ans{present}.
\item $\op{del}(x)$: if $x\in S$ the state becomes $S\setminus\{x\}$ and the
answer is \ans{deleted}; otherwise the state is unchanged and the answer is
\ans{absent}.
\item $\op{mem}(x)$: the answer is \ans{present} if $x\in S$, else
\ans{absent}.
\item $\op{pred}(x)$: the answer is $\max\{s\in S: s<x\}$ if that set is
nonempty, else \ans{none}.
\item $\op{succ}(x)$: the answer is $\min\{s\in S: s>x\}$ if that set is
nonempty, else \ans{none}.
\item $\op{min}$ (resp.\ $\op{max}$): the answer is $\min S$ (resp.\
$\max S$) if $S\neq\emptyset$, else \ans{empty}.
\end{itemize}
No operation other than a fresh insertion or an effective deletion changes
the state.
\end{definition}

A \emph{session} of length $T$ is a sequence of operations
$o_1,\dots,o_T$ starting from $S_0=\emptyset$; the states
$S_0,S_1,\dots,S_T$ and the correct answers are determined by the sequence.
We write $n_t=|S_t|$ and $n=n_T$.

\subsection{Parties, streams, words}

Three parties appear.  The \emph{environment} issues operations.  The
\emph{maintainer} $\Maint$ holds whatever private data it likes, announces an
answer for every operation, and appends to a public append-only sequence,
the \emph{tally}, a short record per operation called a \emph{notch}.  The
\emph{auditor} $\Audit$ passively reads the merged stream of triples
$(o_t,\textit{answer}_t,\textit{notch}_t)$ and keeps a private state; it
sends no messages during the session.  At the audit the maintainer appends a
final notch sequence, the \emph{disclosure}, and the auditor outputs
\textsf{accept} or \textsf{reject}.

We work in the word model: a \emph{word} holds an integer of magnitude
polynomial in $U+T$, and arithmetic on words costs unit time.  All notch
fields, codes, and field elements below fit in $O(1)$ words; with the
concrete prime of Remark~\ref{rem:prime} they fit in one $64$-bit word each.

For soundness the adversary is given every power short of clairvoyance: a
\emph{maintainer strategy} may choose the operations themselves, the answers,
the notches, and the disclosure, adaptively as a function of the entire
visible history.  The auditor's only secret is one field element $z$, drawn
uniformly before the session and never used in any visible action until the
final verdict; consequently the realized transcript is a random variable
independent of $z$.  (Mid-session, the auditor's flag can change only through
checks that do not mention $z$; this independence is used, and re-examined for
chained audits, in Sections~\ref{sec:audit} and~\ref{sec:epochs}.)

\begin{definition}[Correct transcript; completeness; soundness]
\label{def:soundness}
A transcript is \emph{correct} if every announced answer equals the correct
answer of Definition~\ref{def:semantics} for the realized operation prefix.
An auditing scheme is \emph{perfectly complete} if there is a maintainer
strategy---the \emph{honest} one---that always announces correct answers and
is accepted with probability one for every operation sequence and every value
of the auditor's randomness.  The scheme is \emph{$\varepsilon$-sound} if for
every maintainer strategy,
\[
\Pr\bigl[\Audit\ \text{accepts}\ \wedge\ \text{the transcript is incorrect}
\bigr]\ \le\ \varepsilon ,
\]
the probability over the auditor's randomness alone.
\end{definition}

Detection in this model is \emph{retrospective}: an incorrect answer is
exposed at the audit covering it, not at the instant it is announced.
Section~\ref{sec:epochs} shows audits may be placed at arbitrary checkpoints,
so the exposure latency is a parameter, not a fixed cost;
Section~\ref{sec:conclusion} records why the latency cannot be driven to zero
in this information-theoretic setting.

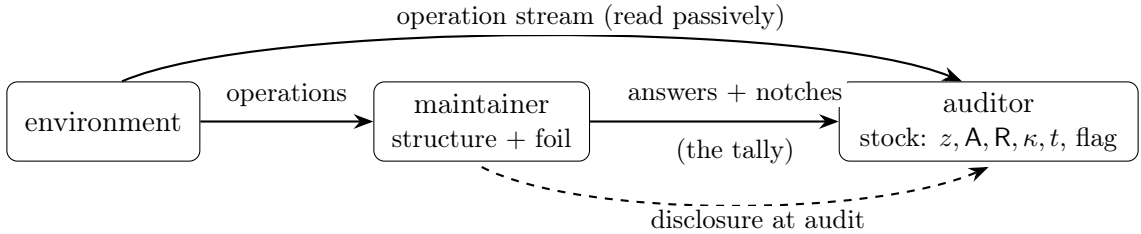
\begin{figure}[H]
\centering
\begin{tikzpicture}[
  envbox/.style={draw, rounded corners, minimum width=2.55cm, minimum height=1.05cm, align=center, fill=white},
  party/.style={draw, rounded corners, minimum width=2.85cm, minimum height=1.05cm, align=center, fill=white},
  audbox/.style={draw, rounded corners, minimum width=3.95cm, minimum height=1.05cm, align=center, fill=white},
  flow/.style={-{Stealth[length=2.6mm]}, thick},
  lbl/.style={font=\small, align=center, fill=white, inner sep=1.5pt}
]
\node[envbox] (env)  at (0,0)    {environment};
\node[party]  (mnt)  at (5.0,0)  {maintainer\\[-1pt]\small structure $+$ foil};
\node[audbox] (aud)  at (11.7,0) {auditor\\[-1pt]\small stock: $z,\acc,\rel,\kap,t$, flag};
\draw[flow] (env.east) -- node[lbl,above=4pt]{operations} (mnt.west);
\draw[flow] (mnt.east) -- (aud.west);
\node[lbl] at (8.35,0.39) {answers $+$ notches};
\node[lbl] at (8.35,-0.39) {(the tally)};
\draw[flow, dashed] ([yshift=-2pt]mnt.south) .. controls (6.3,-1.28) and (10.15,-1.28) ..
  node[lbl,below,pos=.55]{disclosure at audit} ([yshift=-2pt]aud.south);
\draw[flow] ([xshift=7pt]env.north) .. controls (2.1,1.34) and (9.4,1.34) ..
  node[lbl,above,pos=.55]{operation stream (read passively)} ([xshift=-10pt]aud.north);
\end{tikzpicture}
\caption{The split tally.  The maintainer notches the public tally; the
auditor holds the stock; the audit is the meshing of the halves.  No message
ever flows from the auditor.}
\label{fig:model}
\end{figure}

\section{The gap calculus}\label{sec:gaps}

Order semantics becomes checkable once the set is replaced by the family of
its maximal vacant intervals.

\begin{definition}[Gaps]\label{def:gaps}
For a state $S\subseteq\uni$ let $\bar S=S\cup\{0,U+1\}$.  A \emph{gap} of
$S$ is a pair $\gap{a}{b}$ with $a,b\in\bar S$, $a<b$, and
$(a,b)\cap S=\emptyset$; equivalently, $a$ and $b$ are consecutive in
$\bar S$.  Write $\Gaps(S)$ for the set of gaps of $S$.
\end{definition}

\begin{definition}[Interlocking chain]\label{def:chain}
A finite set $G$ of pairs over $\{0,\dots,U+1\}$ is an \emph{interlocking
chain} if, listing its pairs as
$\gap{a_1}{b_1},\dots,\gap{a_k}{b_k}$ in increasing order of left endpoint,
one has $a_1=0$, $b_i=a_{i+1}$ for $1\le i<k$, $b_k=U+1$, and $a_i<b_i$ for
all $i$.  We then say the pairs \emph{mesh}.
\end{definition}

\begin{lemma}[Tiling]\label{lem:tiling}
\,(i) $|\Gaps(S)|=|S|+1$ for every state $S$.
(ii) A set of pairs equals $\Gaps(S)$ for some state $S$ if and only if it is
an interlocking chain.
(iii) The map $S\mapsto\Gaps(S)$ is a bijection from states to interlocking
chains; its inverse reads off the interior endpoints:
$S=\{a_2,\dots,a_k\}=\{b_1,\dots,b_{k-1}\}$.
\end{lemma}

\begin{proof}
List $\bar S$ as $0=s_0<s_1<\dots<s_n<s_{n+1}=U+1$ where $n=|S|$.  By
Definition~\ref{def:gaps}, $\Gaps(S)=\{\gap{s_i}{s_{i+1}}:0\le i\le n\}$,
which has $n+1$ members and is an interlocking chain; this proves (i) and one
direction of (ii).  Conversely, an interlocking chain
$\gap{a_1}{b_1},\dots,\gap{a_k}{b_k}$ determines
$S=\{b_1,\dots,b_{k-1}\}\subseteq\uni$ (interior endpoints are neither $0$
nor $U+1$ because endpoints strictly increase along the chain), and the chain
of consecutive pairs of $\bar S$ is exactly the given one.  The two
constructions invert each other, giving (iii).
\end{proof}

\begin{lemma}[Rewrites]\label{lem:rewrites}
Let $S$ be a state and $x\in\uni$.
\begin{enumerate}[label=(\roman*),leftmargin=2.4em,itemsep=0.1em]
\item If $x\notin S$ there is a unique gap $\gap{a}{b}\in\Gaps(S)$ with
$a<x<b$, and
$\Gaps(S\cup\{x\})=\bigl(\Gaps(S)\setminus\{\gap{a}{b}\}\bigr)
\cup\{\gap{a}{x},\gap{x}{b}\}$.
\item If $x\in S$ there are unique gaps $\gap{a}{x},\gap{x}{b}\in\Gaps(S)$,
and
$\Gaps(S\setminus\{x\})=\bigl(\Gaps(S)\setminus\{\gap{a}{x},\gap{x}{b}\}\bigr)
\cup\{\gap{a}{b}\}$.
\end{enumerate}
\end{lemma}

\begin{proof}
Immediate from the description $\Gaps(S)=\{\gap{s_i}{s_{i+1}}\}$ in the proof
of Lemma~\ref{lem:tiling}: adjoining $x$ refines exactly the consecutive pair
that straddles it; removing $x$ coalesces exactly the two pairs incident to
it.
\end{proof}

\begin{lemma}[Query characterizations]\label{lem:queries}
Let $S$ be a state and $x\in\uni$.
\begin{enumerate}[label=(\alph*),leftmargin=2.4em,itemsep=0.1em]
\item $x\in S$ iff some gap of $S$ has left endpoint $x$, iff some gap of $S$
has right endpoint $x$.  $x\notin S$ iff some gap $\gap{a}{b}$ of $S$ has
$a<x<b$.
\item There is exactly one gap $\gap{a}{b}\in\Gaps(S)$ with $a<x\le b$, and
for it $a=\max\{s\in S:s<x\}$ when that set is nonempty, while $a=0$ exactly
when it is empty.  Symmetrically, exactly one gap satisfies $a\le x<b$, and
its right endpoint is the correct answer of $\op{succ}(x)$, with $b=U+1$
exactly when no successor exists.
\item Exactly one gap has left endpoint $0$, namely $\gap{0}{\min S}$ if
$S\neq\emptyset$ and $\gap{0}{U+1}$ otherwise; symmetrically for right
endpoint $U+1$ and $\max S$.
\end{enumerate}
\end{lemma}

\begin{proof}
All three parts read off the sorted listing
$0=s_0<\dots<s_{n+1}=U+1$.  For (b): the left endpoints
$s_0<s_1<\dots<s_n$ are distinct, and $s_i<x\le s_{i+1}$ holds for exactly
one $i$, namely $i=\max\{j:s_j<x\}$; then $s_i$ is the largest member of
$\bar S$ below $x$, which is the largest member of $S$ below $x$ unless that
set is empty, in which case $s_i=0$.  Part (a) restates consecutiveness:
an interior point of a gap is in no state, and an element of $S$ is the right
endpoint of one gap and the left endpoint of the next.  Part (c) is the case
$x$ at the boundary: $s_0=0$ begins exactly one gap and $s_{n+1}=U+1$ ends
exactly one.
\end{proof}

Lemma~\ref{lem:queries} is the pivot of the whole construction: every answer
of Definition~\ref{def:semantics} is the forced image of a single gap, and by
uniqueness the maintainer will have no latitude in \emph{which} gap to
exhibit.  Lemma~\ref{lem:rewrites} bounds the traffic.

\begin{corollary}[Traffic]\label{cor:traffic}
Under the discipline of Section~\ref{sec:discipline}, every operation
consumes at most two gaps and creates at most two; over a session of $T$
operations, at most $2T+1$ gaps are ever created (including the initial gap
of the empty state) and at most $2T$ are consumed before the audit.
\end{corollary}

\section{Combinatorial normal forms}\label{sec:normalforms}

The preceding lemmas already determine the protocol, but the discrete
structure behind them is useful to state explicitly.  The gaps of a set are
not an arbitrary collection of intervals: they are a unit path in a finite
acyclic interval graph.  This section records the path normal form and the
uniqueness of the two update rewrites.  These statements are the part of the
paper that belongs most directly to discrete mathematics; the field
fingerprint of Section~\ref{sec:discipline} is only a compact way to audit
this finite path calculus.

\begin{definition}[Interval graph and boundary]\label{def:intervalgraph}
Let $\I_U$ be the directed acyclic graph with vertex set
$\{0,1,\dots,U+1\}$ and one directed edge $a\to b$ for every $a<b$.
We identify the edge $a\to b$ with the interval $\gap{a}{b}$.  For a finite
set $H$ of such edges define its discrete boundary
\[
  \partial H(v)=\outdeg_H(v)-\indeg_H(v),
  \qquad v=0,1,\dots,U+1 .
\]
The support $\supp(H)$ is the set of endpoints incident to some edge of $H$.
\end{definition}

\begin{proposition}[Incidence characterization of chains]\label{prop:incidence}
For a finite set $H$ of intervals over $\{0,\dots,U+1\}$, the following are
equivalent.
\begin{enumerate}[label=(\roman*),leftmargin=2.4em,itemsep=0.1em]
\item $H$ is an interlocking chain from $0$ to $U+1$.
\item $H$ is a $\{0,1\}$-valued unit flow in $\I_U$ with boundary
\[
  \partial H=\mathbf{1}_{0}-\mathbf{1}_{U+1} .
\]
Equivalently, $\partial H(0)=1$, $\partial H(U+1)=-1$, and
$\partial H(v)=0$ for every interior vertex $v$.
\end{enumerate}
Consequently, checking that a disclosed family meshes is the same as checking
one finite incidence-vector identity in the interval acyclic graph.
\end{proposition}

\begin{proof}
If $H=\{\gap{s_i}{s_{i+1}}:0\le i\le k-1\}$ with
$0=s_0<\cdots<s_k=U+1$, then $0$ has one outgoing and no incoming edge,
$U+1$ has one incoming and no outgoing edge, and every intermediate $s_i$ has
one incoming and one outgoing edge; vertices not in the support have neither.
Thus (i) implies (ii).

Conversely assume (ii).  Regard the indicator of $H$ as an integral flow of
value one from $0$ to $U+1$.  Because all edges of $\I_U$ point from smaller
to larger vertices, the induced graph is acyclic.  Starting at $0$, there is
one outgoing unit of flow; following any positive-flow edge and using
conservation at each interior vertex must eventually reach $U+1$, since an
acyclic walk cannot continue forever and no other vertex has negative
boundary.  This gives a directed $0$-to-$U+1$ path $P$ contained in $H$.
Subtract the indicator of $P$ from the indicator of $H$.  The remaining edge
set, if nonempty, would be a nonnegative integral circulation in an acyclic
directed graph: its boundary is zero at every vertex.  But any nonempty
finite acyclic directed graph has a smallest vertex incident to an edge with
no incoming remaining edge; at that vertex the remaining boundary is positive,
a contradiction.  Hence no edge remains, so $H=P$, an interlocking chain.
\end{proof}

\begin{theorem}[One-vertex normal form]\label{thm:normalform}
Let $C$ be an interlocking chain and let $V(C)$ be its ordered endpoint set.
\begin{enumerate}[label=(\alph*),leftmargin=2.4em,itemsep=0.1em]
\item If $x\notin V(C)$, there is a unique interlocking chain with endpoint
set $V(C)\cup\{x\}$: if $\gap{a}{b}$ is the unique edge of $C$ with
$a<x<b$, it is
\[
  C^+_x=(C\setminus\{\gap{a}{b}\})\cup
  \{\gap{a}{x},\gap{x}{b}\}.
\]
\item If $x\in V(C)\setminus\{0,U+1\}$, there is a unique interlocking
chain with endpoint set $V(C)\setminus\{x\}$: if
$\gap{a}{x},\gap{x}{b}$ are the two edges of $C$ incident to $x$, it is
\[
  C^-_x=(C\setminus\{\gap{a}{x},\gap{x}{b}\})\cup
  \{\gap{a}{b}\}.
\]
\end{enumerate}
Thus insertion and deletion are not merely local operations that happen to
work; they are the only possible one-key refinements and coarsenings of the
chain normal form.
\end{theorem}

\begin{proof}
Write the endpoint set of $C$ as
$0=s_0<s_1<\cdots<s_k=U+1$.  If $x\notin V(C)$, there is a unique index
$i$ with $s_i<x<s_{i+1}$; any chain whose endpoints are
$V(C)\cup\{x\}$ must connect consecutive endpoints in their sorted order,
so the only changed edge is
$\gap{s_i}{s_{i+1}}$, replaced by
$\gap{s_i}{x}$ and $\gap{x}{s_{i+1}}$.  The deletion case is the same
argument with $x=s_i$ for some $0<i<k$: consecutive endpoints on the two
sides of $x$ must become consecutive after $x$ is removed.  Uniqueness also
follows from Proposition~\ref{prop:incidence}, because the prescribed
endpoint set admits only one unit path.
\end{proof}

\begin{corollary}[Uniqueness of local witnesses]\label{cor:witnessuniq}
For every state $S$ and every operation of Definition~\ref{def:semantics},
the gap or gap pair that can justify the correct branch is unique.  Moreover,
a state-preserving query can only cite the unique chain edge selected by the
corresponding inequality in Lemma~\ref{lem:queries}; an effective update can
only perform the refinement or coarsening of Theorem~\ref{thm:normalform}.
\end{corollary}

\begin{proof}
By Lemma~\ref{lem:tiling}, $\Gaps(S)$ is the unit path whose interior
vertices are exactly $S$.  Membership in $S$ is therefore endpoint membership
in that path, absence is strict containment in one path edge, predecessor and
successor are the two one-sided containment queries of the same ordered path,
and extrema are the two sentinel-anchored edges.  The uniqueness statements
are exactly Lemma~\ref{lem:queries} plus the path-refinement normal form of
Theorem~\ref{thm:normalform}.
\end{proof}

\FloatBarrier
\section{The indenture discipline}\label{sec:discipline}

The auditor cannot afford to remember the gap family, so it remembers a
fingerprint of its \emph{history}: every gap, when created, is stamped with a
fresh code; every operation must \emph{cite} the codes of the gaps it
rewrites or reads; and the two sides of the ledger---births and
consumptions---are folded into two field accumulators.  This section defines
the discipline purely behaviorally; no structure that might implement it is
mentioned, and none needs to be.  Any maintainer whose notches obey the rules
is treated identically.

\subsection{Codes and the public clock}

The auditor keeps a counter $\kap$, initially $0$.  Whenever a gap is
created---we say \emph{born}---it receives the current value of $\kap$ as its
\emph{code}, and $\kap$ increments.  Before the first operation, the gap
$\gap{0}{U+1}$ of the empty state is born with code $0$.  Within an event,
all cited codes are checked first, against the value $\kap_{\mathrm{start}}$
that $\kap$ had when the event began; births then occur in the fixed order
prescribed by the matching rule below.  An \emph{indentured gap} is a triple
$\igap{a}{b}{c}$: a gap together with its code.

\begin{lemma}[Public clock]\label{lem:clock}
The code assigned to every birth is a deterministic function of the public
transcript prefix.  Consequently any party reading the tally---in particular
a maintainer---can reproduce all code assignments exactly, with $O(1)$
bookkeeping per event and no communication from the auditor.
\end{lemma}

\begin{proof}
The counter starts at a fixed value, increments only at births, and the
number and order of births in each event are fixed by the operation and the
notch shape, both public.
\end{proof}

\subsection{Encoding}

Codes, endpoints, and gaps are folded into a prime field through one
injective map.

\begin{definition}[Encoding]\label{def:encoding}
Fix a prime $p$.  For $a,b\in\{0,\dots,U+1\}$ and a code $c\ge 0$ define
\[
\nuf(a,b,c)\;=\;1+c\,(U+2)^2+a\,(U+2)+b \pmod p .
\]
\end{definition}

\begin{lemma}[Injectivity]\label{lem:injective}
If $p>(2T+2)(U+2)^2$ then $\nuf$ is injective on
$\{0,\dots,U+1\}^2\times\{0,\dots,2T+1\}$, which contains every triple that
can arise in a session of $T$ operations.
\end{lemma}

\begin{proof}
For $a,b\le U+1$ one has $0\le a(U+2)+b\le (U+1)(U+2)+(U+1)=(U+1)(U+3)<
(U+2)^2$, so the integer $1+c(U+2)^2+a(U+2)+b$ determines $c$ as a quotient
and $(a,b)$ as a remainder: distinct triples give distinct integers.  All
such integers lie in $[1,(2T+2)(U+2)^2]\subseteq[1,p-1]$, so reduction
modulo $p$ identifies nothing.  Corollary~\ref{cor:traffic} bounds codes by
$2T+1$.
\end{proof}

\begin{remark}[A concrete instantiation]\label{rem:prime}
On a $64$-bit machine take $p=2^{61}-1$.  Then
Lemma~\ref{lem:injective} holds whenever $(2T+2)(U+2)^2<2^{61}-1$, e.g.\ for
every $U\le 10^6$ and $T\le 10^6$ simultaneously, and every quantity in the
scheme fits in one word with single-word modular arithmetic.  The choice is
parametric: nothing below depends on it beyond the inequality.
\end{remark}

\subsection{The stock and the notch grammar}

\begin{definition}[Stock]\label{def:stock}
The auditor's private state, the \emph{stock}, consists of five words and a
flag: a uniformly random secret $z\in\F_p$ drawn before the session; two
accumulators $\acc,\rel\in\F_p$, both initialized to $1$; the public counter
$\kap$; the event index $t$; and a one-bit flag, initially \textsf{ok}.
On the birth of $\igap{a}{b}{c}$ the auditor updates
$\acc\leftarrow\acc\cdot(z-\nuf(a,b,c))$;
on the consumption of $\igap{a}{b}{c}$ it updates
$\rel\leftarrow\rel\cdot(z-\nuf(a,b,c))$.
\end{definition}

Each operation must be accompanied by a notch matching one of the rules of
Table~\ref{tab:grammar}.  Each rule prescribes the notch's fields, a local
predicate the auditor evaluates, the citations consumed, the births (in
order), and the unique answer the citation forces; the auditor rejects---%
permanently lowering the flag---on any malformed notch, failed predicate,
violated temporal rule, or announced answer differing from the forced one.
The operation's key $x$ is part of the operation itself, so left-anchored
citations \rl{2}, \rl{5} and the extremal citations \rl{8}, \rl{8$'$} omit
the implied endpoint, which is how membership confirmations cost two words.

\begin{table}[H]
\centering\small
\begin{tabular}{@{}llllll@{}}
\toprule
Rule & Operation, branch & Notch & Predicate & Consumes & Births (order)\\
\midrule
\rl{1} & $\op{ins}(x)$, fresh   & $\langle a,b,c\rangle$ & $a<x<b$ & $\igap{a}{b}{c}$ & $\gap{a}{x}$, $\gap{x}{b}$\\
\rl{2} & $\op{ins}(x)$, duplicate & $\langle b,c\rangle$ & $x<b$ & $\igap{x}{b}{c}$ & $\gap{x}{b}$\\
\rl{3} & $\op{del}(x)$, present & $\langle a,c_1,b,c_2\rangle$ & $a<x<b$ & $\igap{a}{x}{c_1},\,\igap{x}{b}{c_2}$ & $\gap{a}{b}$\\
\rl{4} & $\op{del}(x)$, absent  & $\langle a,b,c\rangle$ & $a<x<b$ & $\igap{a}{b}{c}$ & $\gap{a}{b}$\\
\rl{5} & $\op{mem}(x)$, present & $\langle b,c\rangle$ & $x<b$ & $\igap{x}{b}{c}$ & $\gap{x}{b}$\\
\rl{6} & $\op{mem}(x)$, absent  & $\langle a,b,c\rangle$ & $a<x<b$ & $\igap{a}{b}{c}$ & $\gap{a}{b}$\\
\rl{7} & $\op{pred}(x)$ & $\langle a,b,c\rangle$ & $a<x\le b$ & $\igap{a}{b}{c}$ & $\gap{a}{b}$\\
\rl{7$'$} & $\op{succ}(x)$ & $\langle a,b,c\rangle$ & $a\le x<b$ & $\igap{a}{b}{c}$ & $\gap{a}{b}$\\
\rl{8} & $\op{min}$ & $\langle b,c\rangle$ & --- & $\igap{0}{b}{c}$ & $\gap{0}{b}$\\
\rl{8$'$} & $\op{max}$ & $\langle a,c\rangle$ & --- & $\igap{a}{U{+}1}{c}$ & $\gap{a}{U{+}1}$\\
\bottomrule
\end{tabular}
\caption{The notch grammar.  Forced answers: \rl{1} \ans{inserted};
\rl{2}, \rl{5} \ans{present}; \rl{3} \ans{deleted}; \rl{4}, \rl{6}
\ans{absent}; \rl{7} answer $a$, read as \ans{none} when $a=0$; \rl{7$'$}
answer $b$, read as \ans{none} when $b=U+1$; \rl{8} answer $b$, read as
\ans{empty} when $b=U+1$; \rl{8$'$} answer $a$, read as \ans{empty} when
$a=0$.  Every cited or born endpoint pair must satisfy
$0\le a<b\le U+1$.}
\label{tab:grammar}
\end{table}

\begin{definition}[Temporal rule \rl{0}]\label{def:temporal}
Every cited code must satisfy $c<\kap_{\mathrm{start}}$, the counter value at
the beginning of the event.  In words: one may cite only what was already
issued before this event began.
\end{definition}

\begin{definition}[Disclosure \rl{9}]\label{def:disclosure}
At the audit, the maintainer appends a sequence of triples
$\igap{a_1}{b_1}{c_1},\dots,\igap{a_k}{b_k}{c_k}$.  The auditor rejects
unless: $k\le\kap$ (the \emph{length cap}); the pairs mesh
(Definition~\ref{def:chain}); and every $c_i<\kap$.  Each disclosed triple is
consumed into $\rel$.  The auditor then \textsf{accept}s iff the flag is
\textsf{ok} and $\acc=\rel$.
\end{definition}

\begin{figure}[H]
\centering
\fbox{\begin{minipage}{0.93\textwidth}\small
\textbf{Auditor's event routine} (state: $z,\acc,\rel,\kap,t$, flag).
On reading $(o,\textit{answer},\textit{notch})$, if the flag is lowered, do
nothing; else:
\begin{enumerate}[leftmargin=2em,itemsep=0.05em]
\item $t\leftarrow t+1$; \ $\kap_{\mathrm{start}}\leftarrow\kap$.
\item Match the notch against the rule for $o$ in
Table~\ref{tab:grammar}; on shape mismatch, lower the flag and stop.
\item Check the rule's predicate and the endpoint bounds
$0\le a<b\le U+1$ of every citation; on failure, lower the flag.
\item Check the temporal rule: every cited code $c<\kap_{\mathrm{start}}$;
on failure, lower the flag.
\item For each citation $\igap{a}{b}{c}$, in the rule's order:
$\rel\leftarrow\rel\cdot(z-\nuf(a,b,c))$.
\item For each birth $\gap{a}{b}$, in the rule's order:
$\acc\leftarrow\acc\cdot(z-\nuf(a,b,\kap))$; \ $\kap\leftarrow\kap+1$.
\item Compute the forced answer from the citation
(Table~\ref{tab:grammar}); if it differs from \textit{answer}, lower the
flag.
\end{enumerate}
\textbf{At the audit}: apply Definition~\ref{def:disclosure};
output \textsf{accept} iff the flag is \textsf{ok} and $\acc=\rel$.
\end{minipage}}
\caption{The auditor, in full.  Steps 1--4 and 7 never read $z$; only steps
5--6 and the final comparison do.}
\label{fig:auditor}
\end{figure}

Figure~\ref{fig:auditor} assembles the routine.  Two structural remarks, both
used later.  First, the flag is a function of the public transcript alone:
$z$ enters only the accumulators and the final equality, so nothing the
auditor could be observed doing mid-session depends on $z$.  Second, the
discipline never asks the auditor to know the state: every check is local to
one event, and the auditor's entire memory of history is the pair
$(\acc,\rel)$, the counter, and the index.

\begin{definition}[Well-notched; balanced]\label{def:wellnotched}
A transcript is \emph{well-notched} if it lowers no flag: every notch matches
its rule, every predicate and endpoint bound holds, the temporal rule holds,
every answer equals the forced answer, and the disclosure passes the length
cap, meshing, and code-bound checks.  For a well-notched transcript let
$\Births$ be the multiset of indentured gaps born (including the initial
$\igap{0}{U+1}{0}$) and $\Consumed$ the multiset consumed---during the run
and at disclosure.  The transcript is \emph{balanced} if
$\Births=\Consumed$ as multisets.
\end{definition}

\section{The audit theorem}\label{sec:audit}

Two deterministic lemmas carry the weight: balance forces every citation to
be \emph{live}, and live citations force every answer to be correct.  The
field then converts balance into one probable equality.

\begin{definition}[Ledger]\label{def:ledger}
Fix a well-notched transcript.  For $0\le t\le T$ let $\Live_t$ be the
multiset of indentured gaps born in events $0,\dots,t$ (event $0$ being the
initial birth) minus those consumed in events $1,\dots,t$.  A citation at
event $t$ is \emph{live} if the cited triple belongs to $\Live_{t-1}$.
\end{definition}

\begin{lemma}[Liveness]\label{lem:liveness}
In a well-notched, balanced transcript:
\begin{enumerate}[label=(\alph*),leftmargin=2.4em,itemsep=0.1em]
\item $\Births$ is a set: no two births share a code, and each code
determines its gap.
\item $\Consumed=\Births$ is likewise a set; in particular no indentured gap
is consumed twice, during the run or at disclosure.
\item Every citation in the run is live, and the disclosure equals
$\Live_T$ as a set.
\end{enumerate}
\end{lemma}

\begin{proof}
(a) Codes are assigned by a strictly increasing counter, one per birth, and
the gap of a birth is fixed at the moment the code is issued.  (b) Balance
says the multisets are equal; a multiset equal to a set is a set.  Thus each
code appears exactly once on each side, and the matched pairs carry the same
gap by (a).

(c) Consider a citation of $\igap{a}{b}{c}$ at event $t$, during the run or
at disclosure.  By (b) there is exactly one birth with code $c$, and it is the
birth of the very gap $\gap{a}{b}$.  By the temporal rule
(during the run) or the disclosure code bound (at the audit), the birth of
code $c$ occurred strictly before the citing event began, so the triple
entered the ledger before event $t$.  By (b) it is consumed exactly once in
the whole transcript, namely here; hence it had not been consumed earlier and
$\igap{a}{b}{c}\in\Live_{t-1}$.  For the final claim:
$\Live_T=\Births-\Consumed_{\mathrm{run}}$, and balance gives
$\Births-\Consumed_{\mathrm{run}}=\Consumed_{\mathrm{disc}}$, the disclosed
multiset.
\end{proof}

\begin{lemma}[Forcing]\label{lem:forcing}
In a well-notched, balanced transcript, for every $0\le t\le T$ the gaps of
the ledger are exactly the gaps of the true state,
\[
\{\gap{a}{b}:\igap{a}{b}{c}\in\Live_t\}\;=\;\Gaps(S_t)
\quad\text{(with all multiplicities equal to one),}
\]
and every announced answer in the transcript is correct.
\end{lemma}

\begin{proof}
Induction on $t$.  At $t=0$, $\Live_0=\{\igap{0}{U+1}{0}\}$ and
$\Gaps(\emptyset)=\{\gap{0}{U+1}\}$.  Assume the claim after event $t-1$ and
consider event $t$ with operation $o_t$ on key $x$ and true prior state
$S=S_{t-1}$.  By Lemma~\ref{lem:liveness}(c) every gap cited at event $t$
lies in $\Live_{t-1}$, hence by induction in $\Gaps(S)$.

\emph{Branch forcing.}  Suppose $o_t=\op{ins}(x)$ and the notch matched
\rl{2}: the citation is a gap of $S$ with left endpoint $x$, so
$x\in S$ by Lemma~\ref{lem:queries}(a); the forced answer \ans{present} is
correct, the state is unchanged, and the rewrite---consume $\gap{x}{b}$,
rebirth $\gap{x}{b}$---leaves the ledger's gap family equal to
$\Gaps(S)=\Gaps(S_t)$.  If instead the notch matched \rl{1}: the citation is
a gap of $S$ with $a<x<b$, so $x\notin S$ by the same lemma (a gap contains
no element of $S$ strictly inside); the forced answer \ans{inserted} is
correct, and the rewrite implements Lemma~\ref{lem:rewrites}(i) exactly, so
the ledger's family becomes $\Gaps(S\cup\{x\})=\Gaps(S_t)$.  The maintainer
cannot choose the wrong branch: a live gap with left endpoint $x$ exists iff
$x\in S$, and a live gap strictly containing $x$ exists iff $x\notin S$.
Multiplicity one is preserved because cited triples leave the ledger and
births carry fresh codes.

$\op{del}(x)$ is symmetric through rules \rl{3}, \rl{4} and
Lemma~\ref{lem:rewrites}(ii): rule \rl{3} can only be satisfied by citing the
two gaps incident to $x$, which exist iff $x\in S$ and are unique by
Lemma~\ref{lem:rewrites}(ii); rule \rl{4} requires a containing gap, which
exists iff $x\notin S$.  $\op{mem}(x)$ uses \rl{5}, \rl{6} with the same
dichotomy and a state-preserving rewrite.

For $\op{pred}(x)$ under \rl{7}, the citation is a live gap with
$a<x\le b$; by Lemma~\ref{lem:queries}(b) exactly one gap of $S$ satisfies
this, and its left endpoint $a$ is precisely the correct answer, with $a=0$
exactly in the \ans{none} case---which is the forced reading.  The rewrite is
state-preserving.  $\op{succ}$, $\op{min}$, $\op{max}$ are identical through
Lemma~\ref{lem:queries}(b),(c); for the extrema the cited gap is the unique
one anchored at the relevant sentinel, and the sentinel value of the free
endpoint encodes emptiness.

Thus the answer at event $t$ is correct and the ledger again mirrors
$\Gaps(S_t)$, completing the induction; ranging over $t$ gives the second
claim.
\end{proof}

\begin{theorem}[Audit]\label{thm:audit}
Fix $U$, $T$, and a prime $p>(2T+2)(U+2)^2$.  The scheme of
Sections~\ref{sec:discipline}--\ref{sec:audit} is perfectly complete, and it
is $\varepsilon$-sound with
\[
\varepsilon\;\le\;\frac{4T+1}{p}
\]
against every maintainer strategy, with no bound on the maintainer's
computational power.  The auditor stores five words and a flag and performs
$O(1)$ word operations per event and per disclosed triple.
\end{theorem}

\begin{proof}
\emph{Completeness.}  The honest maintainer maintains the true state, cites
at each event the unique gap or gap pair that Lemmas~\ref{lem:rewrites}
and~\ref{lem:queries} prescribe, with the codes those gaps actually carry---%
reproducible by Lemma~\ref{lem:clock}---and at the audit discloses
$\Gaps(S_T)$ with its live codes.  Every predicate of
Table~\ref{tab:grammar} holds by construction, every announced answer is the
forced one, the temporal rule holds because cited gaps were born at earlier
events, the disclosure meshes by Lemma~\ref{lem:tiling}(ii) and has length
$n_T+1\le\kap$.  The transcript is balanced---every birth is eventually
consumed, in the run or at disclosure---so $\acc$ and $\rel$ are products
over identical multisets and agree for \emph{every} $z$.  Acceptance is
certain; no randomness is consumed by completeness.

\emph{Soundness.}  Fix a maintainer strategy.  Since the auditor emits
nothing during the session and its flag depends only on the public
transcript, the realized transcript $\tau$ is independent of $z$.  Condition
on $\tau$.  If $\tau$ is not well-notched, the flag is lowered and the
auditor rejects with probability one.  So assume $\tau$ well-notched, with
birth multiset $\Births$ and consumption multiset $\Consumed$ determined by
$\tau$.  Define the polynomial
\[
P(Z)\;=\;\prod_{\beta\in\Births}\bigl(Z-\nuf(\beta)\bigr)\;-\;
\prod_{\gamma\in\Consumed}\bigl(Z-\nuf(\gamma)\bigr)\;\in\;\F_p[Z],
\]
so that the auditor accepts iff $P(z)=0$.

If $\Births=\Consumed$ then $P\equiv 0$ and the auditor accepts; but by
Lemma~\ref{lem:forcing} every answer in $\tau$ is then correct, so this case
contributes nothing to the soundness event.  If $\Births\neq\Consumed$, then
by Lemma~\ref{lem:injective} the two products have different multisets of
roots in $\F_p$; both being products of monic linear factors, unique
factorization in $\F_p[Z]$ gives $P\not\equiv 0$.  Its degree is at most
$\max(|\Births|,|\Consumed|)$.  By Corollary~\ref{cor:traffic},
$|\Births|\le 2T+1$; consumptions during the run number at most $2T$ and the
disclosure is capped at $\kap\le 2T+1$ triples
(Definition~\ref{def:disclosure}), so $|\Consumed|\le 4T+1$.  A nonzero
polynomial of degree at most $4T+1$ has at most $4T+1$ roots, and $z$ is
uniform on $\F_p$ and independent of $\tau$, hence
\[
\Pr_z\bigl[\text{accept}\mid\tau\bigr]
=\Pr_z\bigl[P(z)=0\mid\tau\bigr]\le\frac{4T+1}{p}.
\]
Averaging over $\tau$ in the event that the transcript is incorrect (which,
as shown, forces $\Births\neq\Consumed$ whenever the flag survived) yields
the bound.  The cost accounting is immediate from
Figure~\ref{fig:auditor}.
\end{proof}

\begin{remark}[Calibration]\label{rem:calibration}
For a soundness target $2^{-\lambda}$ it suffices to take
$p\ge(4T+1)\,2^{\lambda}$, compatible with Remark~\ref{rem:prime} for all
practical $T$; e.g.\ $p=2^{61}-1$ gives error below $2^{-38}$ for
$T=10^6$.  The bound degrades only linearly in session length, so a single
audit may cover long histories.
\end{remark}

\begin{corollary}[Meshing]\label{cor:meshing}
Condition on acceptance and exclude the $(4T+1)/p$ failure event.  Then the
disclosure equals $\Gaps(S_T)$ together with its live codes: the foil
meshes with the truth necessarily, not merely syntactically.  The explicit
meshing check of Definition~\ref{def:disclosure} is therefore redundant for
soundness; it is retained because it rejects deterministically, without
reference to $z$, and because Section~\ref{sec:epochs} re-births the
disclosed triples and needs them to form a chain even in the failure event.
\end{corollary}

\begin{proof}
Outside the failure event, acceptance implies balance, and
Lemmas~\ref{lem:liveness}(c) and~\ref{lem:forcing} identify the disclosure
with $\Live_T$ and its gap family with $\Gaps(S_T)$.
\end{proof}

\section{Necessity}\label{sec:necessity}

Each pillar of Theorem~\ref{thm:audit}---the randomness, its secrecy, and the
temporal rule---is now shown to be irreplaceable.  For the lower bounds we
allow the auditor an arbitrary notch language: the maintainer may append any
finite binary string per event, and the auditor is any machine over such
streams.

\begin{definition}[Deterministic auditor]\label{def:detauditor}
A \emph{deterministic auditor with $s$ bits of state} is a deterministic
automaton that reads the stream of (operation, answer, notch-string) records,
updates a state of $s$ bits per record, and outputs \textsf{accept} or
\textsf{reject} at an end-of-session audit as a function of its final state
and the audit records.  It is \emph{perfectly complete} if some maintainer
strategy always announces correct answers and is always accepted; it is
\emph{sound} if it never accepts an incorrect transcript.
\end{definition}

\begin{theorem}[Deterministic floor]\label{thm:floor}
Let $n\ge 1$ and $U=2n$.  Every deterministic auditor that is perfectly
complete and sound for sessions over $\uni=\{1,\dots,2n\}$ of length at
least $n+1$ satisfies
\[
s\;\ge\;\log_2\binom{2n}{n}\;\ge\;2n-\log_2(2n+1).
\]
\end{theorem}

\begin{proof}
Fix the honest strategy $\Maint_0$ witnessing perfect completeness.  For
each $n$-element subset $X\subseteq\uni$ let $\pi_X$ denote the record
stream produced when the environment issues
$\op{ins}(x_1),\dots,\op{ins}(x_n)$, the elements of $X$ in increasing
order, and $\Maint_0$ answers and notches; let $\sigma(X)\in\{0,1\}^s$ be
the auditor's state after reading $\pi_X$.  Note first that the auditor
cannot have rejected during $\pi_X$: the stream extends to a complete honest
session (append a correct disclosure under $\Maint_0$), which perfect
completeness obliges the auditor to accept, and rejection is absorbing.

Suppose $2^s<\binom{2n}{n}$.  By pigeonhole there are distinct $n$-sets
$X\neq Y$ with $\sigma(X)=\sigma(Y)$.  Both have size $n$, so there exists
$x^\ast\in X\setminus Y$.  Consider the honest continuation in world $Y$:
the environment issues $\op{mem}(x^\ast)$, $\Maint_0$ answers
\ans{absent}---correct, since $x^\ast\notin Y$---with its honest notch, the
audit is called, and $\Maint_0$ produces its honest disclosure for $Y$.
Write $\rho$ for this entire suffix of records, byte for byte.  Started from
state $\sigma(Y)$, the auditor accepts $\pi_Y\rho$ by perfect completeness.

Now splice.  A cheating maintainer in world $X$ plays $\pi_X$ honestly, then
replays the byte string $\rho$: it announces \ans{absent} to
$\op{mem}(x^\ast)$ and copies world $Y$'s notch and disclosure verbatim.
The auditor is deterministic; from the merge point onward its state evolution
is a function of (current state, remaining bytes), and both are identical to
the world-$Y$ run---$\sigma(X)=\sigma(Y)$ and the suffixes coincide.  Hence
it accepts.  But in the spliced transcript the operation prefix is
$\op{ins}$ of all of $X$ followed by $\op{mem}(x^\ast)$ with
$x^\ast\in X$: the correct answer is \ans{present}, the announced answer is
\ans{absent}, and an incorrect transcript has been accepted, contradicting
soundness.  Therefore $2^s\ge\binom{2n}{n}$, and the numeric bound follows
from $\binom{2n}{n}\ge 4^{\,n}/(2n+1)$.
\end{proof}

\begin{corollary}[Visible coins]\label{cor:visible}
Let an auditor toss coins $\rho$ but reveal them to the maintainer before the
session (or, equivalently, act so that the maintainer can infer them).
Suppose it is perfectly complete---for every $\rho$, every honest session is
accepted---and keeps $s<\log_2\binom{2n}{n}$ bits of state beyond $\rho$.
Then its soundness error is $1$: some maintainer strategy is accepted with
probability one while announcing an incorrect answer.
\end{corollary}

\begin{proof}
Fix any value of $\rho$.  The auditor with $\rho$ hard-wired is a
deterministic auditor with $s$ bits of state that is perfectly complete, so
by the proof of Theorem~\ref{thm:floor} there exist, \emph{for this}
$\rho$, sets $X_\rho\neq Y_\rho$ and a spliced transcript
$\pi_{X_\rho}\rho'$ that it accepts despite an incorrect answer.  The
maintainer, who sees $\rho$, selects and plays that very transcript.
Acceptance holds for every $\rho$, hence with probability one.
\end{proof}

Corollary~\ref{cor:visible} delimits the design space exactly: below the
$\binom{2n}{n}$ floor, an auditor must randomize, must hide the randomness,
and---if it is to stay perfectly complete, as ours is---must hide it
unconditionally.  The five-word stock of Theorem~\ref{thm:audit} sits as low
in that space as the floor permits, with one secret word doing all the
hiding.

\subsection{The temporal rule cannot be dropped}

The remaining pillar is rule \rl{0}.  Its necessity is sharper than a
counting bound: removing it admits a forgery that is accepted not with
noticeable probability but with certainty, because the forged transcript is
\emph{exactly} balanced.  The public clock, which is what makes the honest
maintainer's bookkeeping free (Lemma~\ref{lem:clock}), is also what a forger
exploits: future codes are predictable.

\begin{proposition}[Temporal replay]\label{prop:replay}
Let the \emph{weakened auditor} be the auditor of Figure~\ref{fig:auditor}
with step 4 deleted and all other checks intact, including meshing and the
disclosure code bound.  There is a maintainer strategy over $U=8$ issuing
three operations that announces an incorrect answer and is accepted by the
weakened auditor with probability one, for every $z$.  The full auditor
rejects the same transcript deterministically at its second event.
\end{proposition}

\begin{proof}
The initial gap is $\igap{0}{9}{0}$ and $\kap=1$.  The strategy issues
$\op{ins}(5)$, $\op{mem}(7)$, $\op{ins}(7)$, then calls the audit.

\emph{Event 1}, $\op{ins}(5)$, rule \rl{1}, notch $\langle 0,9,0\rangle$:
consume $\igap{0}{9}{0}$; births $\igap{0}{5}{1}$, $\igap{5}{9}{2}$;
$\kap=3$.  Honest so far; the state is $\{5\}$.

\emph{Event 2}, $\op{mem}(7)$, announced \ans{present}---incorrect, since
$7\notin\{5\}$---rule \rl{5}, notch $\langle 9,5\rangle$: the citation is
$\igap{7}{9}{5}$.  Code $5$ has not been issued
($\kap_{\mathrm{start}}=3$); the weakened auditor does not notice, finds the
predicate $7<9$ satisfied, consumes $\igap{7}{9}{5}$, and rebirths
$\gap{7}{9}$ as $\igap{7}{9}{3}$; $\kap=4$.  The forger could write code $5$
because, by Lemma~\ref{lem:clock}, it can compute today what the counter
will issue tomorrow.

\emph{Event 3}, $\op{ins}(7)$, rule \rl{1}, notch $\langle 5,9,2\rangle$:
consume $\igap{5}{9}{2}$---genuinely live since event 1---and birth
$\igap{5}{7}{4}$ and $\igap{7}{9}{5}$; $\kap=6$.  The second birth receives
precisely the code cited one event earlier.

\emph{Disclosure}: $\igap{0}{5}{1}$, $\igap{5}{7}{4}$, $\igap{7}{9}{3}$.
The pairs mesh ($0\to5\to7\to9$), the length $3\le\kap=6$, and all codes are
below $6$, so every disclosure check passes.  Now compare the ledgers as
multisets:
\[
\Births=\bigl\{\igap{0}{9}{0},\igap{0}{5}{1},\igap{5}{9}{2},
\igap{7}{9}{3},\igap{5}{7}{4},\igap{7}{9}{5}\bigr\},
\]
\[
\Consumed=\bigl\{\igap{0}{9}{0},\igap{7}{9}{5},\igap{5}{9}{2}\bigr\}
\cup\bigl\{\igap{0}{5}{1},\igap{5}{7}{4},\igap{7}{9}{3}\bigr\}
=\Births .
\]
Balance is exact, so $\acc=\rel$ as polynomials in $z$ and the weakened
auditor accepts for every $z$, despite the incorrect answer at event 2.
The full auditor, by contrast, rejects at event 2 in step 4: the cited code
$5$ is not below $\kap_{\mathrm{start}}=3$, a deterministic, $z$-free
rejection.
\end{proof}

\begin{remark}[Meshing is no substitute]\label{rem:nomesh}
The disclosure in the forgery is a perfect interlocking chain---indeed it
\emph{is} $\Gaps(\{5,7\})$ with plausible codes---so no end-of-session
consistency condition on the disclosed family can detect the attack.  What
rule \rl{0} polices is not the final shape of the ledger but the
\emph{direction of time} inside it: a citation must point backward.  The
proposition shows that with citations allowed to point forward, the two
halves of a forged tally can be carved to mesh exactly.
\end{remark}

\section{A maintainer, and its audited rebalancing}\label{sec:maintainer}

Nothing in Sections~\ref{sec:discipline}--\ref{sec:necessity} mentions an
implementation, and nothing needs to: the discipline is satisfied or it is
not, whatever produces the notches.  This section shows the discipline costs
an honest implementation essentially nothing, and then turns the
implementation's own running time into one more auditable claim.

\subsection{A leaf-oriented \texorpdfstring{$(2,4)$}{(2,4)}-tree}

\begin{theorem}[Maintainer]\label{thm:maintainer}
There is a maintainer that, for every operation sequence, announces correct
answers, emits notches satisfying the indenture discipline, and produces the
honest disclosure, with the following costs: $O(\log n)$ worst-case time per
operation, $O(1)$ words of notch per operation, $n+O(1)$ words of storage
beyond the keys, and $O(n)$ time for the disclosure.  In particular the
scheme of Theorem~\ref{thm:audit} is perfectly complete at no asymptotic
overhead to the structure.
\end{theorem}

\begin{proof}
Keep the elements of $S$ in the leaves of a $(2,4)$-tree: every internal
node has between $2$ and $4$ children, all leaves are at the same depth,
and the keys appear in the leaves in increasing
order~\cite{bayermccreight72,ahu74,huddlestonmehlhorn82}; internal nodes
store routing guides (the maximum key of each subtree).  Thread the leaves
into a doubly linked list and prepend a permanent head cell representing the
sentinel $0$; keep pointers to the head and the last leaf.  The \emph{foil}
is one word per cell: the head and each leaf $\ell$ store the code of the gap
whose left endpoint is the key of that cell---that is, the gap from this key
to its successor (to $U+1$ at the last leaf, and $\gap{0}{U+1}$ at the head
when $S=\emptyset$).  By Lemma~\ref{lem:tiling} this stores every live code
exactly once.

Each operation first locates $x$ by a root-to-leaf search in $O(\log n)$
time; the predecessor and successor cells are then available in $O(1)$
through the leaf links (for $\op{min}$, $\op{max}$, through the head and
tail pointers, in $O(1)$ total).  The notch of Table~\ref{tab:grammar} is
assembled by reading at most two stored codes; the foil is then updated to
the codes the public clock will assign, which the maintainer computes by
Lemma~\ref{lem:clock} from its own mirror of $\kap$: a fresh insertion
writes $\kap$ into the predecessor's cell and $\kap+1$ into the new leaf; a
deletion writes $\kap$ into the predecessor's cell and discards the removed
leaf's; every query rewrites the single cited cell to $\kap$.  Rebalancing
---splits, fusions, shares of internal nodes---moves leaves but never
inspects or alters a stored code: codes live in the cells and travel with
them.  Hence each operation costs $O(\log n)$ worst-case time (search plus a
root-to-leaf rebalancing walk) and $O(1)$ work for the discipline.  The
disclosure walks the leaf list once, emitting
$\igap{\text{key}}{\text{next}}{\text{code}}$ from the head onward: $n+1$
triples in $O(n)$ time, meshing by construction.  Correct answers,
satisfied predicates, the temporal rule, and balance are exactly the honest
behavior verified in the completeness half of Theorem~\ref{thm:audit}.
Storage: one code word per element, plus the head, plus the mirror of
$\kap$.
\end{proof}

\subsection{The attestation envelope}

The maintainer of Theorem~\ref{thm:maintainer} does a variable amount of
internal work per operation: a single insertion can cascade splits to the
root.  Classically one proves such structures do $O(1)$ rebalancing per
update \emph{on average}~\cite{huddlestonmehlhorn82,tarjan85}, and the claim
lives in the analysis, invisible at run time.  Here the claim can be moved
into the tally.  Call a split, a fusion, a share, a root creation, or a root
collapse a \emph{structural event}.

\begin{theorem}[Attestation envelope]\label{thm:envelope}
Over any sequence of $m$ update operations ($\op{ins}$ and $\op{del}$,
effective or not) applied to the initially empty maintainer of
Theorem~\ref{thm:maintainer}, the total number of structural events is at
most $2m$; moreover the bound holds at every prefix of the sequence.
\end{theorem}

\begin{proof}
Assign to each internal node of degree $d$ the potential
\[
\varphi(1)=3,\qquad \varphi(2)=1,\qquad \varphi(3)=0,\qquad
\varphi(4)=2,\qquad \varphi(5)=4,
\]
degrees $1$ and $5$ occurring only transiently, and let $\Phi$ be the sum
over all internal nodes; $\Phi=0$ for the empty structure and $\Phi\ge 0$
always.  Decompose each update into elementary steps and write the
\emph{amortized cost} of a step as its number of structural events plus the
change in $\Phi$.  We verify every step:

\smallskip
\noindent\emph{Insertions.}
\begin{itemize}[leftmargin=2em,itemsep=0.1em]
\item Creating the first leaf: no internal node, no event, $\Delta\Phi=0$;
amortized $0$.
\item Creating the first internal root (second insertion), degree $2$: one
root-creation event, $\Delta\Phi=\varphi(2)=+1$; amortized $2$.
\item Arrival of a leaf (or, in a cascade, of a split-off node) at a node of
degree $d\in\{2,3,4\}$: no event,
$\Delta\Phi\in\{\varphi(3)-\varphi(2),\varphi(4)-\varphi(3),
\varphi(5)-\varphi(4)\}=\{-1,+2,+2\}\le+2$; amortized $\le 2$.
\item Split of a non-root node of degree $5$ into degrees $2$ and $3$, its
parent gaining a child: one event,
$\Delta\Phi=\bigl(\varphi(2)+\varphi(3)-\varphi(5)\bigr)
+\Delta\varphi_{\text{parent}}\le-3+2=-1$; amortized $\le 0$.
\item Split of a degree-$5$ root: one split and one root creation (the new
root has degree $2$), $\Delta\Phi=-3+\varphi(2)=-2$; amortized $0$.
\end{itemize}

\noindent\emph{Deletions.}
\begin{itemize}[leftmargin=2em,itemsep=0.1em]
\item Removing the only leaf: no internal node, amortized $0$.
\item Departure of a leaf (or of a fused-away node) from a node of degree
$d\in\{2,3,4\}$: no event,
$\Delta\Phi\in\{+2,+1,-2\}\le+2$; amortized $\le 2$.
\item Share: a node of degree $1$ borrows a child from an adjacent sibling
of degree $e\in\{3,4\}$: one event,
$\Delta\Phi=\bigl(\varphi(2)-\varphi(1)\bigr)
+\bigl(\varphi(e-1)-\varphi(e)\bigr)\le-2+1=-1$; amortized $\le 0$, and the
cascade ends, since the parent's degree is unchanged.
\item Fusion: a node of degree $1$ merges into an adjacent sibling of degree
$2$, leaving one node of degree $3$; the parent loses a child: one event,
$\Delta\Phi=\bigl(\varphi(3)-\varphi(1)-\varphi(2)\bigr)
+\Delta\varphi_{\text{parent}}\le-4+2=-2$; amortized $\le-1$, and the
cascade, if it continues at the parent, is already paid for.
\item Root collapse: the root reaches degree $1$ (through one of the steps
above, which charged its $+2$) and is removed, its child promoted: one
event, $\Delta\Phi=-\varphi(1)=-3$; amortized $-2$.
\end{itemize}

Each update performs exactly one arrival or departure step (the only steps
of positive amortized cost, each $\le 2$), or is a no-op; all other steps
have nonpositive amortized cost.  Summing over any prefix of $m$ updates,
\[
\#\{\text{structural events}\}\;=\;\sum(\text{amortized costs})
-\bigl(\Phi_{\text{final}}-\Phi_0\bigr)\;\le\;2m-\Phi_{\text{final}}
\;\le\;2m. \qedhere
\]
\end{proof}

The constants are not an artifact of generosity: a workload that inserts
$m/2$ consecutive keys and deletes them again drives the event count to
$m-O(\log m)$, so the envelope is tight within a factor of two, and no
potential of this form can beat the inherent one-event-per-update floor of
alternating split--fuse churn.

\begin{definition}[Attested mode]\label{def:modeii}
In \emph{attested mode} the maintainer appends one extra word to the tally
per structural event---an \emph{attestation} naming the event's kind---and
the auditor keeps two additional counters, of updates seen and attestations
seen, rejecting the moment attestations exceed twice updates.
\end{definition}

\begin{corollary}[Audited amortization]\label{cor:modeii}
In attested mode the honest maintainer is never rejected
(Theorem~\ref{thm:envelope} holds at every prefix), and any accepted
transcript carries at most $2m$ attestations.  The amortized-cost claim of
Theorem~\ref{thm:envelope} is thereby checked by the auditor itself, by
counting against a public envelope, rather than asserted by analysis alone.
The certificate's scope is exactly the attested schedule: a maintainer that
performs structural work without attesting it gains nothing except a slower
structure, since attestations carry no semantic weight in
Theorem~\ref{thm:audit} and omitting one cannot create or hide a gap.
\end{corollary}

\begin{corollary}[Session tally]\label{cor:tally}
By Table~\ref{tab:grammar}, an update notches at most $4$ words and a query
at most $3$; the disclosure is $3(n_T+1)$ words; attested mode adds at most
$2m$ words.  A session of $m$ updates and $q$ queries ending at size $n_T$
therefore writes at most
\[
4m+3q+2m+3(n_T+1)
\]
words of tally.  For instance, $m=600$ updates and $q=400$ queries ending at
$n_T=100$ write at most
$4\cdot600+3\cdot400+2\cdot600+3\cdot101=5103$ words---about forty kilobytes
on a $64$-bit machine---against which the auditor holds five words and a
flag, six in attested mode.
\end{corollary}

\section{Checkpoints and epochs}\label{sec:epochs}

The audit of Theorem~\ref{thm:audit} need not wait for the end of the
maintainer's life.  An audit may be called at any instant $t$, and an
accepted audit may open the next \emph{epoch} in place, reusing the same
secret.

\begin{definition}[Epoch protocol]\label{def:epochs}
Audits are called at instants $t_1<t_2<\cdots<t_k$, splitting the session
into epochs $i=1,\dots,k$ of lengths $T_i=t_i-t_{i-1}$ (with $t_0=0$).  At
$t_i$ the maintainer discloses as in Definition~\ref{def:disclosure}, except
that the length cap for epoch $i$ is $k_{i-1}+2T_i$, where $k_{i-1}$ is the
length of the previous accepted disclosure ($k_0=1$, the initial gap).  The
auditor verifies the epoch as before and announces the verdict;
\emph{on reject it halts permanently}.  On accept it opens epoch $i+1$:
it resets $\acc\leftarrow1$, $\rel\leftarrow1$, and re-births each disclosed
triple's gap with a fresh code from the continuing counter $\kap$, the
maintainer mirroring the assignment by Lemma~\ref{lem:clock}.
\end{definition}

\begin{theorem}[Checkpoints and epochs]\label{thm:epochs}
Under the epoch protocol with a prime $p$ exceeding
$(2\kap_{\mathrm{max}})(U+2)^2$ for the session's final counter value
$\kap_{\mathrm{max}}$:
\begin{enumerate}[label=(\alph*),leftmargin=2.4em,itemsep=0.1em]
\item each audit costs one pass over its disclosure, $O(k_{i-1}+T_i)$ field
operations in $O(1)$ words of working space, and the honest maintainer is
accepted at every audit with probability one;
\item the probability that some accepted epoch contains an incorrect answer
is at most
\[
\sum_{i=1}^{k}\frac{4T_i+n_{t_{i-1}}+1}{p},
\]
where $n_{t_{i-1}}$ is the true size at the start of epoch $i$.
\end{enumerate}
\end{theorem}

\begin{proof}
(a) Costs and completeness are as in Theorem~\ref{thm:audit}, applied per
epoch; the honest disclosure of epoch $i$ has length $n_{t_i}+1\le
k_{i-1}+T_i$ within the cap, and the re-birth seeds epoch $i+1$ with the
true gap family carrying fresh codes, so the honest invariants are restored
verbatim.

(b) The verdict of each epoch is announced, so unlike the single-audit
setting the maintainer's later behavior may depend on information correlated
with $z$.  Halt-on-reject collapses this dependence: epoch $i$ is reached
only along the all-accept history, and along that one history the
maintainer's strategy determines the epoch-$i$ transcript $\tau_i$ as a
fixed object, independent of $z$.  Let $B_i$ be the event that $z$ reaches
epoch $i$ and $\tau_i$ is well-notched with unequal birth and consumption
multisets yet $P_i(z)=0$, where $P_i$ is the difference polynomial of
$\tau_i$ as in Theorem~\ref{thm:audit}.  Then
$\Pr[B_i]\le\Pr_z[P_i(z)=0]\le\deg P_i/p$, and
$\deg P_i\le\max\bigl(|\Births_i|,|\Consumed_i|\bigr)$ with
$|\Births_i|\le k_{i-1}+2T_i$ (the re-births plus at most two per event) and
$|\Consumed_i|\le 2T_i+(k_{i-1}+2T_i)$ by the epoch cap, hence
$\deg P_i\le 4T_i+k_{i-1}$.

On the complement of $B_1\cup\dots\cup B_k$, induct on $i$: epoch $i$ starts
with the re-births of an accepted, balanced epoch $i-1$, whose disclosure
equals $\Gaps(S_{t_{i-1}})$ with multiplicity one by
Corollary~\ref{cor:meshing}; the liveness and forcing lemmas
(\ref{lem:liveness}, \ref{lem:forcing}) then apply within epoch $i$ with
this family as the base case---the meshing check guarantees the re-birthed
family is a chain in every case---so an accepted epoch $i$ is balanced and
all its answers are correct, and $k_i=n_{t_i}+1$.  Therefore every accepted
epoch on this event is correct, and
\[
\Pr\bigl[\text{some accepted epoch is incorrect}\bigr]
\le\sum_i\Pr[B_i]\le\sum_i\frac{4T_i+k_{i-1}}{p}
=\sum_i\frac{4T_i+n_{t_{i-1}}+1}{p}.
\]
The prime bound keeps Lemma~\ref{lem:injective} valid for every code the
session ever issues.
\end{proof}

Epochs make the exposure latency of Section~\ref{sec:model} a tunable: audit
every $T_i$ operations and a false answer survives at most $T_i$ further
operations, at $O(n+T_i)$ amortized audit cost and additive error.  Taken to
the extreme $T_i=1$ the protocol audits after every operation; the error sum
grows as $T(n+5)/p$ and the per-operation cost as $O(n)$---the right regime
for small structures under intense suspicion, and a poor one otherwise,
which is the trade the model prices explicitly.

\section{An audited session, end to end}\label{sec:exemplar}

Every quantity in this section is illustrative data, not a design choice:
the scheme's parameters remain $p$, $U$, $T$ throughout the paper, and the
numbers here merely instantiate them small enough to recompute by hand.
Take $U=6$ (sentinels $0$ and $7$), a session of $T=6$ operations, and the
prime $p=997$: Lemma~\ref{lem:injective} requires
$p>(2T+2)(U+2)^2=14\cdot64=896$, which $997$ clears and which, for the
record, $101$ would not.  Since $a(U+2)+b\le 8\cdot7+7=63<64$, the encoding
reads $\nuf(a,b,c)=1+64c+8a+b$, and with ten codes its largest value is
$1+64\cdot9+8\cdot6+7=632<997$, so no reduction ever occurs and every
residue below can be checked by hand.  The auditor's secret is $z=131$; in
deployment $z$ is uniform and hidden, and $131$ is shown only so the reader
can multiply along.

The session: $\op{ins}(4)$, $\op{ins}(2)$, $\op{mem}(3)$, $\op{ins}(3)$,
$\op{del}(2)$, $\op{max}$.  The true states run
$\emptyset$, $\{4\}$, $\{2,4\}$, $\{2,4\}$, $\{2,3,4\}$, $\{3,4\}$, $\{3,4\}$, so the correct
answers are \ans{inserted}, \ans{inserted}, \ans{absent}, \ans{inserted},
\ans{deleted}, $4$.  Table~\ref{tab:honest} runs the honest maintainer.
Event 4 displays the touch-and-rebirth mechanism earning its keep: the
membership test at event 3 consumed $\igap{2}{4}{4}$ and rebirthed the same
gap as $\igap{2}{4}{5}$, so the insertion of $3$ one event later must cite
the \emph{reborn} code $5$---citing the spent code $4$ would unbalance the
ledger.

\begin{table}[H]
\centering\footnotesize\setlength{\tabcolsep}{3.5pt}
\begin{tabular}{@{}clllrr@{}}
\toprule
$t$ & event, rule, notch & consumes ($\nuf$) & births ($\nuf$) & $\acc$ & $\rel$\\
\midrule
0 & initial birth & --- & $\igap{0}{7}{0}$ (8) & 123 & 1\\
1 & $\op{ins}(4)$ \ans{inserted}, \rl{1} $\langle0,7,0\rangle$ &
$\igap{0}{7}{0}$ (8) & $\igap{0}{4}{1}$ (69), $\igap{4}{7}{2}$ (168) & 986 & 123\\
2 & $\op{ins}(2)$ \ans{inserted}, \rl{1} $\langle0,4,1\rangle$ &
$\igap{0}{4}{1}$ (69) & $\igap{0}{2}{3}$ (195), $\igap{2}{4}{4}$ (277) & 904 & 647\\
3 & $\op{mem}(3)$ \ans{absent}, \rl{6} $\langle2,4,4\rangle$ &
$\igap{2}{4}{4}$ (277) & $\igap{2}{4}{5}$ (341) & 587 & 253\\
4 & $\op{ins}(3)$ \ans{inserted}, \rl{1} $\langle2,4,5\rangle$ &
$\igap{2}{4}{5}$ (341) & $\igap{2}{3}{6}$ (404), $\igap{3}{4}{7}$ (477) & 685 & 708\\
5 & $\op{del}(2)$ \ans{deleted}, \rl{3} $\langle0,3,3,6\rangle$ &
$\igap{0}{2}{3}$ (195), $\igap{2}{3}{6}$ (404) & $\igap{0}{3}{8}$ (516) & 480 & 397\\
6 & $\op{max}$ answer $4$, \rl{8$'$} $\langle4,2\rangle$ &
$\igap{4}{7}{2}$ (168) & $\igap{4}{7}{9}$ (616) & 498 & 266\\
\midrule
 & disclose $\igap{0}{3}{8}$ (516) & & & 498 & 281\\
 & disclose $\igap{3}{4}{7}$ (477) & & & 498 & 480\\
 & disclose $\igap{4}{7}{9}$ (616) & & & 498 & 498\\
\bottomrule
\end{tabular}
\caption{The honest session over $U=6$, $p=997$, $z=131$.  Parenthesized
values are $\nuf(a,b,c)=1+64c+8a+b$; the accumulator columns show the
running products modulo $997$.  The disclosure
$\gap{0}{3},\gap{3}{4},\gap{4}{7}$ is $\Gaps(\{3,4\})$ and meshes.  Final
comparison: $\acc=\rel=498$, \textsf{accept}---as Theorem~\ref{thm:audit}
guarantees for every $z$, not only this one.}
\label{tab:honest}
\end{table}

\begin{table}[H]
\centering\footnotesize\setlength{\tabcolsep}{3.5pt}
\begin{tabular}{@{}clllrr@{}}
\toprule
$t$ & event, rule, notch & consumes ($\nuf$) & births ($\nuf$) & $\acc$ & $\rel$\\
\midrule
0 & initial birth & --- & $\igap{0}{7}{0}$ (8) & 123 & 1\\
1 & $\op{ins}(4)$ \ans{inserted}, \rl{1} $\langle0,7,0\rangle$ &
$\igap{0}{7}{0}$ (8) & $\igap{0}{4}{1}$ (69), $\igap{4}{7}{2}$ (168) & 986 & 123\\
2 & $\op{ins}(2)$ \ans{inserted}, \rl{1} $\langle0,4,1\rangle$ &
$\igap{0}{4}{1}$ (69) & $\igap{0}{2}{3}$ (195), $\igap{2}{4}{4}$ (277) & 904 & 647\\
3 & $\op{mem}(3)$ \textbf{\ans{present}}, \rl{5} $\langle4,4\rangle$ &
$\igap{3}{4}{4}$ (285) & $\igap{3}{4}{5}$ (349) & 334 & 62\\
4 & $\op{ins}(3)$ \ans{inserted}, \rl{1} $\langle2,4,5\rangle$ &
$\igap{2}{4}{5}$ (341) & $\igap{2}{3}{6}$ (404), $\igap{3}{4}{7}$ (477) & 901 & 938\\
5 & $\op{del}(2)$ \ans{deleted}, \rl{3} $\langle0,3,3,6\rangle$ &
$\igap{0}{2}{3}$ (195), $\igap{2}{3}{6}$ (404) & $\igap{0}{3}{8}$ (516) & 71 & 50\\
6 & $\op{max}$ answer $4$, \rl{8$'$} $\langle4,2\rangle$ &
$\igap{4}{7}{2}$ (168) & $\igap{4}{7}{9}$ (616) & 460 & 144\\
\midrule
 & disclose $\igap{0}{3}{8}$ (516) & & & 460 & 392\\
 & disclose $\igap{3}{4}{7}$ (477) & & & 460 & 957\\
 & disclose $\igap{4}{7}{9}$ (616) & & & 460 & 457\\
\bottomrule
\end{tabular}
\caption{The forged twin: identical except at event 3, where the maintainer
announces \ans{present} for the absent key $3$, fabricating the citation
$\igap{3}{4}{4}$.  Code $4$ is below $\kap_{\mathrm{start}}=5$, the
predicate $3<4$ holds, and the rebirth is well-formed, so every local check
of the \emph{full} auditor passes; the lie is invisible until the meshing of
the halves.  Final comparison: $\acc=460\neq 457=\rel$,
\textsf{reject}.}
\label{tab:forged}
\end{table}

Table~\ref{tab:forged} runs the forged twin.  The cheat is as gentle as
possible: one answer is flipped, the fabricated citation is shape-perfect
and temporally legal (code $4$ was genuinely issued before event 3---to a
\emph{different} gap), and every subsequent event copies the honest run.
The two phantom factors it creates---the consumption of
$\igap{3}{4}{4}$, never born, and of $\igap{2}{4}{5}$, whose code in this
run went to the rebirth of $\gap{3}{4}$---remain unmatched against the
orphaned births $\igap{2}{4}{4}$ and $\igap{3}{4}{5}$, and the audit fails:
$460\neq457$ at $z=131$.

The instance also illustrates how loose the worst-case bound of
Theorem~\ref{thm:audit} typically is.  The difference polynomial of the twin
is $Q(Z)\cdot\bigl[(Z-277)(Z-349)-(Z-285)(Z-341)\bigr]$, where $Q$ collects
the sixteen factors the two ledgers share; since
$277+349=285+341=626$, the bracket is the nonzero constant
$277\cdot349-285\cdot341\equiv 485\pmod{997}$, so this particular forgery
is accepted only if $z$ happens to equal one of the at most eight shared
encodings---probability at most $8/997$, far below the generic
$(4T+1)/p=25/997$, and zero at $z=131$.

For contrast with Proposition~\ref{prop:replay}: the forgery of
Table~\ref{tab:forged} respects the temporal rule and is caught by balance;
the replay forgery respects balance and is caught by the temporal rule.
The two checks fail independently, and Sections~\ref{sec:audit}
and~\ref{sec:necessity} prove that together they leave no third route.

\section{Concluding remarks}\label{sec:conclusion}

The results form a closed account of the constant-memory audit of dynamic
ordered sets.  The gap calculus (Lemmas~\ref{lem:tiling}--\ref{lem:queries})
and its incidence normal form (Proposition~\ref{prop:incidence} and
Theorem~\ref{thm:normalform}) turn order semantics into a finite path
rewriting system whose every answer is the forced image of one citation.  The
indenture discipline prices the certificate: at most four words of notch per update, three per query, three
per element at disclosure, against a stock of five words and a flag.  The
audit theorem (Theorem~\ref{thm:audit}) delivers perfect completeness and
soundness error $(4T+1)/p$ against unbounded adaptive maintainers; the
calibration is explicit, and a $61$-bit Mersenne prime covers a million
operations over a million keys with error below $2^{-38}$.  The necessity
results close the design space from below: $\binom{2n}{n}$ states for any
deterministic auditor (Theorem~\ref{thm:floor}), the same wall for visible
coins (Corollary~\ref{cor:visible}), and certain forgery without the
temporal rule (Proposition~\ref{prop:replay}).  The maintainer's theorem and
the attestation envelope (Theorems~\ref{thm:maintainer}
and~\ref{thm:envelope}) settle the other side of the ledger: the discipline
costs the structure one word per element and $O(1)$ work per operation, and
even its amortized rebalancing budget---at most $2m$ structural events---can
be policed by the auditor's counter rather than taken on faith.  Epoch
composition (Theorem~\ref{thm:epochs}) makes audit frequency a dial with
additive error.

Two boundary lines deserve plain statement.  First, detection here is
retrospective by design, and the necessity results say something about how
far that can be improved: an auditor below the $\binom{2n}{n}$ floor must
keep its randomness hidden, and whether such an auditor can be made
\emph{online}---rejecting a false answer the moment it is announced rather
than at the covering audit---is not a gap in the present analysis but a
recognized frontier: Naor and Rothblum~\cite{naorrothblum09} proved that
online memory checking with small secret state and nontrivial query cost
implies the existence of one-way functions, with the quantitative trade-offs
mapped in~\cite{dnrv09}.  Within the information-theoretic setting of this
paper, retrospective is therefore the correct word, and
Theorem~\ref{thm:epochs} prices the latency exactly.  Second, the scope is
the seven-operation interface of Definition~\ref{def:semantics}.  The gap
calculus certifies arrangement---adjacency, bracketing, extremality---and
does so completely; positional statistics such as the rank of a key are
aggregate properties of the arrangement, not local ones, and lie outside
what a two-word citation can force.  The interface boundary is drawn where
the uniqueness arguments of Lemma~\ref{lem:queries} end, and within that
boundary nothing has been left conditional: every constant is explicit,
every protocol total, and both halves of the tally---the construction and
its impossibility counterparts---have been cut from the same stick.

\subsection*{Acknowledgements}
The authors thank their colleagues at the Department of Computer Engineering
for discussions on transcript verification and on the history of accounting
instruments.

\end{document}